\documentclass[preprint,review,12pt,authoryear]{elsarticle}

\usepackage{amssymb}
\usepackage{amsmath}

\usepackage{amssymb}
\usepackage{amsfonts}
\usepackage{amsmath}

\usepackage{graphicx}
\usepackage{epstopdf}
\usepackage{pgf,pgfarrows,pgfnodes,pgfautomata,pgfheaps,pgfplots}
\usepackage{cite}
\usepackage{appendix}
\usepackage{xurl}

\newtheorem{claim}{Claim}

\newtheorem{definition}{Definition}

\newenvironment{proof}[1][Proof]{\noindent\textbf{#1.} }{\ \rule{0.5em}{0.5em}}
\begin{document}

\begin{frontmatter}



\title{Green Transition with Dynamic Social Preferences} 


\author[1]{Kirill Borissov} 

\affiliation[1]{organization={European University at St. Petersburg},
           addressline={6/1A Gagarinskaya Street}, 
           city={St. Petersburg},
           postcode={191187}, 
           country={Russia}}

\author[2]{Nigar Hashimzade} 

\affiliation[2]{organization={Brunel University of London},
           addressline={Kingston Lane}, 
           city={Uxbridge},
           postcode={UB8 3PH}, 
           state={Middlesex},
           country={United Kingdom}}

\begin{abstract}
We examine a green transition policy involving a tax on brown goods in an economy where preferences for green consumption consist of a constant intrinsic individual component and an evolving social component. We analyse equilibrium dynamics when social preferences exert a positive  externality in green consumption, creating complementarity between policy and preferences. The results show that accounting for this externality allows for a lower tax rate compared to policy ignoring the social norm effects. Furthermore, stability conditions permit gradual tax reductions or even removal along the transition path, minimising welfare losses. Thus, incorporating policy-preference interactions improves green transition policy design.
\end{abstract}




\end{frontmatter}

\section{Introduction}
Among the numerous effects of climate change on societies and habitat is a potentially significant reduction in global economic capacity. \citet{burke2015global} predict 23 percent decline in global GDP caused by global warming by 2100 compared to the scenario without climate change. A primary driver of this loss is a decrease in productivity, which, according to this study, peaks at $13$ degrees centigrade and rapidly fall at higher temperatures. Furthermore, human health is an important factor of labour productivity.  The detrimental effect of climate change on human health is one of its least contested and most urgent consequences (\citealp{oecd2024}). Rising temperatures and polluted environment exacerbate infectious and non-communicable diseases (\citealp{who2023noncommunicable}, \citealp{kazmierczak2022climate}) as well as respiratory and cardiovascular conditions (\citealp{giorgini2017climate}), directly impairing workers’ productive capacity and economic output (\citealp{zhang2011measuring}).

It is now widely established that greenhouse gas (GHG) emissions constitute one of the main reasons of global warming. To mitigate this pressing issue, societies must transition from high-emission (`brown') goods and technologies to low-emission (`green') alternatives. A range of environmental policies facilitating this transition has been adopted in many countries, with carbon taxes, in one form or another, frequently at the core of these measures. Equally important, however, is the intrinsic motivation of the population, driven by individual and societal values, to alter their consumption choices. Recently, the literature has turned to analysing the interaction between environmental policies and environmental values. Our paper contributes to this body of research by modelling the effect of social norms on individual preferences in the form of social externalities in green consumption and examining its interaction with the effect of a carbon tax levied on brown consumption. 

The empirical evidence of social norms in environmental preferences was documented by \citet{andre2021fighting}. Based on a survey of a large representative sample of U.S. adult population, the authors find that, firstly, the willingness to combat climate change varies significantly across individuals, and, secondly, it was significantly positively affected by the individual perception of social norms: the willingness was higher if a respondent believed that higher proportion of population is also willing to fight climate change. The effect of the social norm in this context is similar to the externality in consumption of network goods, although the channel is different. This analogy suggests that the presence of social norms in green behaviour can alter the optimal design of climate policies in a way similar to the presence of network consumption externalities. 

The importance of taking into account the network effect in setting carbon taxes was highlighted, for example, by \citet{sartzetakis2005environmental}, who showed that when the network effect benefits dirty technology the Pigouvian carbon tax must be set higher than the marginal environmental damage in order to resolve the negative externality. Similarly, \citet{greaker2016network} considered the situation with network externalities in both the clean and the dirty sectors and showed that a tax which only takes into account the environmental externality leads to excess inertia in green transition. Another example of externalities, on the production side rather than consumption or preferences, in the climate policies context is analysed in the work by \citet{borissov2019carbon}. In their model the externality, generated by the societal human capital, is in labour productivity, and its interaction with the carbon tax introduced to facilitate the switch to the green technology determines the tax design. 

There is also empirical evidence of the effect of climate policies on consumers' preferences. Interestingly, the literature reports both crowding-in and crowding-out effects (see, e.g. \citealp{lanz2018behavioral}, \citealp{bartels2025more}, \citealp{pizzo2025carbon}). \citet{mattauch2022economics} modelled a direct effect of carbon tax on consumers' preference and showed that recognising the endogeneity of preferences with respect to climate policies can improve their design. \citet{besley2023political} model dynamic interaction between green and brown technologies and values when green values contain an intrinsic moral component and are also partly affected by economic incentives and, thus, by policies altering these incentives. The authors focus on the positive analysis of green transition policies in a political economy setting and demonstrate the importance of co-evolution of values and technologies for the characteristics of the political equilibria.

Our work is built on similar premises, centred on the assumption of the social effect in consumers' preferences for green and brown goods. An intrinsic, or moral component of the preferences is combined with the social norm determined by the relative prevalence of green and brown goods consumption in the society. Brown goods are associated with GHG emissions damaging health and thus, reducing the labour productivity and the overall productive capacity of the economy. In this setting we analyse the policy of sales tax on brown good used to fund public healthcare services. The effect of this policy on preferences is mediated by the social norm, as the policy affects the aggregate production and consumption of the green and brown goods both through the relative price of the goods (tax makes the brown good relatively more expensive) and through the households' labour income (public healthcare improves labour productivity). We characterise the dynamic and the steady-state equilibria (SSE) in this economy and derive the conditions for convergence to the green SSE. We show that taking into account endogenous evolution of preferences helps design a gradual policy, with moderate tax that can even be temporary, thus avoiding a large negative shock to consumption and welfare from a high tax that ignores the change in preferences. 

\section{The Model}
The model economy consists of the firms, the households, and the government. The firms specialise in producing one of the two consumption goods, green and brown, that have similar use. These could be, for example, electric cars and petrol cars, or heat pumps and gas boilers, or plant-based meat and dairy substitutes and conventional meat and dairy. The markets are perfectly competitive, the firms operate a constant-returns-to-scale technology and, therefore, the production sector can be described by two representative firms. The households supply labour to the firms and use their labour income to purchase the goods. The brown good contributes to GHG emissions and climate change, which has detrimental effect on health and, therefore, on labour productivity. Thus, the brown good consumption generates a negative externality reducing output and labour income, ignored by private agents. The consumption of green good generates a positive social externality: an observed increase in the aggregate green good consumption strengthens the consumer preferences towards the green good. The government collects taxes and uses tax revenues to fund public healthcare services mitigating the effect of climate change on health. Time is discrete, and all production, consumption, and policy decisions are made in every period.

\subsection{Producers}

The outputs $G_t$ and $B_t$ of the green and brown goods at time $t$ are determined by linear production functions. For simplicity, we assume that labour is the only variable factor of production:
\[
	G_t = a^{g} \mu_t l_t^{g},
\] 
\[
	B_t = a^{b}\mu_t l_t ^{b},
\]	
where the scale factors $a^{g}$ and $a^{b}$ are the fixed inputs or technology levels in the green and brown good sectors, $l_t^{g}$ and $l_t^{b}$ are labour inputs, and $\mu_t$ is labour productivity determined by the health status of the population. The current health status of the population decreases in the quantity of the brown good and increases in the level of public healthcare services in the previous time period, \[ \mu_t = \mu (B_{t-1},H_{t-1}), \] where $B_{t-1}$ is the total output of good $B$ at time $t-1$ and $H_{t-1}$ is the level of public healthcare services at time $t-1$. We assume that $\mu (B,H)$ is a continuous function and set $\mu (0,0) = 1$.

The workers are perfectly mobile across sectors. The outputs of the two goods are determined as the solution to the firms' profit-maximisation problems,
\begin{equation*}\label{pm-c0}
\max_{l_t^{g}\geq 0}	p_tG_t - w_tl_t, \ \text{s.t.} \ G_t = \mu (B_{t-1},H_{t-1})a^{g}l_t^{g} 
\end{equation*}
and
\begin{equation*}\label{pm-d0}
	\max_{l_t^{b}\geq 0}	B_t - w_tl_t, \ \text{s.t.} \ B_t = \mu (B_{t-1},H_{t-1})a^{b}l_t^{b},
\end{equation*}	
where  $p_t$ is the price of the green good at time $t$, $w_t$ is the time-$t$ wage rate, and the price of the brown good is normalised to $1$. 
These problems can be rewritten as follows: 
\begin{equation}\label{pm-c}
	\max_{G_t\geq 0}	p_tG_t - \frac{w_tG_t}{\mu (B_{t-1},H_{t-1})a^{g}}, 
\end{equation}
\begin{equation}\label{pm-d}
	\max_{B_t\geq 0}	B_t - \frac{w_tB_t}{\mu (B_{t-1},H_{t-1})a^{b}}.
\end{equation}

\subsection{Households}	
	There is a continuum $[0,1]$ of households, labelled by index $i$, each endowed with one unit of labour. As workers, the households supply labour in every time period and earn income equal to the wage rate at that time. As consumers, the households spend income on purchases of the brown and the green good. 
	
	We assume that the green and the brown goods are perfect substitutes in consumption. The utility consumer $i$ derives from  consumption of the green and the brown goods, $g_t$ and $b_t$, at time $t$ takes the form 
\[
u^i(g_t, b_t) = \lambda_{t}\gamma(i) g_t + b_t,
\]
	where an intrinsic constant individual-specific relative preference for the green good, denoted by $\gamma(i)$, is augmented by the socially determined preference parameter at time $t$, denoted by $\lambda_{t}$. 
	
	The function $\gamma(i)$ is assumed to be continuous and strictly decreasing (in other words, the consumers are labelled in the decreasing order of their green good preferences).
	
	For the social component of preferences we assume that $\lambda_{t} = \lambda (\frac{G_{t-1}}{B_{t-1}})$, where $G_{t-1}$ and $B_{t-1}$ are the aggregate consumption levels of the green and the brown good at time $t-1$. The function $\lambda (\cdot)$ is continuous and increasing. This reflects the social norm or peer pressure in population preferences: an increase in the overall consumption of the green good relative to the brown good makes the green good more attractive for each consumer.

The brown good purchases are subject to a sales tax. Given the time-$t$ wage rate, $w_t$, the price of the green good, $p_t$, and the rate of tax on consumption of the brown good, $\tau_{t}$, the budget constraint of an individual household is given by
	\[
	p_tg_t + (1+\tau_{t})b_t = w_t,
	\] 
	where $g_t$ and $b_t$ are the individual consumption levels of the green and brown goods at time $t$.
    
	Thus, at time $t$ household $i$ solves the following utility maximization problem:
\begin{equation}\label{um}
	\max \lambda_{t}\gamma(i) g_t + b_t,  \ \text{s.t.} \ 
	p_tg_t + (1+\tau_{t})b_t = w_t,
\end{equation} 
The solution to this problem, $(g_t(i),b_t(i))$, is as follows: 
\begin{equation*}
	g_t(i) = \begin{cases}
		\frac{w_t}{p_t},   &\lambda_{t}\gamma(i) > \frac{p_t}{1+\tau_{t}} \\
		0,  & \lambda_{t}\gamma(i) < \frac{p_t}{1+\tau_{t}} 
	\end{cases}
\end{equation*}
\begin{equation*}
	b_t(i) = \begin{cases}
		0,   &\lambda_{t}\gamma(i) > \frac{p_t}{1+\tau_{t}} \\
		\frac{w_t}{1+\tau_{t}},  & \lambda_{t}\gamma(i) < \frac{p_t}{1+\tau_{t}} 
	\end{cases}
\end{equation*}

Thus, every household purchases either green or brown good. Moreover, there exists $j\in [0,1]$ such that all households with $i \in [0,j)$ choose to buy the green good and all households with $i \in (j,1]$ choose to buy the brown good. Thus, $j$ is the share of households who consume the green good and $1-j$ is the share of households who consume the brown good. The marginal household $j$, characterised by $\lambda_{t}\gamma(j) = \frac{p_t}{1+\tau_{t}}$, is indifferent between the two goods. For determinacy,  we take as the solution for the marginal household either the vector  $(g_t(j),b_t(j))$ given by $g_t(j) = \frac{w_t}{p_t}$ and $b_{t}(j) = 0$ or the vector  $(g_t(j),b_t(j))$ given by $g_t(j) = 0$ and $b_t(j) = \frac{w_t}{1+\tau_{t}}$. Which of these two vectors to take will be clear from the context. 

\subsection{Government}

The government collects the sales tax from purchases of the brown good and uses the tax revenues to fund the public healthcare system. The production of healthcare services is given by \[
H_t = a^h \mu_t l^h_t
\] 

To simplify the notations further on, without loss of generality we set $a^h=a^b$ (this amounts to the choice of the unit of measurement of $H$). We assume that the government runs balanced budget every period, \[H_t = T_t \equiv \tau_t B_t.\] 

\section{Equilibrium and Its Dynamic Properties}

Suppose we are given $B_{t-1}$, $T_{t-1}$, $\lambda_{t}$ and $\tau_{t}$. We define the dynamic equilibrium as the time path of prices and quantities such that in every time period the firms maximise profits, the households maximise utility, the markets for goods and for labour inputs clear, the firms earn zero profits, and the government balances the budget.

\begin{definition}
The tuple $(p_{t},w_{t},G_{t},B_{t},H_t,(g_{t}(i),b_{t}(i))_{i\in [0,1]})$ constitutes a time$-t$ equilibrium if (i) $G_{t}$ solves (\ref{pm-c}); (ii) $B_{t}$ solves (\ref{pm-d}); (iii)  for every $i\in [0,1]$, $(g_{t}(i),b_{t}(i))$ solves (\ref{um}); and (iv) the following equalities hold:
\begin{itemize}
	\item $G_{t} = \int_{i=0}^{1} g_{t}(i) di$;
	\item $B_{t} = \int_{i=0}^{1} b_{t}(i) di$;
   \item $\frac{G_{t}}{\mu (B_{t-1},T_{t-1})a^{g}} + \frac{B_{t}}{\mu (B_{t-1},T_{t-1})a^{b}} + \frac{\tau_t B_t}{\mu (B_{t-1},T_{t-1})a^{b}} = 1$;
	\item $w_{t} = \mu (B_{t-1},H_{t-1})a^{b}$;
     \item $H_t = \tau_t B_t$.
\end{itemize}

\end{definition}

The time$-t$ equilibrium is determined by the share  of households who purchase the green good, i.e. the number $j_{t}$ such that 
\begin{equation}\label{c-eq}
	g_{t}(i) = \begin{cases}
		\frac{w_{t}}{p_{t}},   & i < i^*_{t} \\
		0,  & i > i^*_{t}, 
	\end{cases}
\end{equation}
\begin{equation}\label{d-eq}
	b_{t}(i) = \begin{cases}
		0,   &i < i^*_{t} \\
		\frac{w_{t}}{1+\tau_{t}},  & i > i^*_{t}. 
	\end{cases}
\end{equation}
 
The firms earn zero profits in equilibrium under the perfect competition assumption. Therefore, from (\ref{pm-c}) and (\ref{pm-d}), 
\begin{equation} \label{w-eq}
	w_{t} = p_{t} \mu (B_{t-1},H_{t-1})a^{g} = \mu (B_{t-1},H_{t-1})a^{b},
\end{equation}
\begin{equation}\label{C-eq}
	G_{t} = i^*_{t}\frac{w_{t}}{p_{t}}, 
\end{equation}
\begin{equation}\label{D-eq}
	B_{t} = (1 - i^*_{t})w_{t},
\end{equation}
\begin{equation}\label{p-eq}
	p_{t} = \frac{a^{b}}{a^{g}}.
\end{equation}
\begin{equation}\label{C/D-eq}
	\frac{G_{t}}{B_{t}} 
	= \frac{i^*_{t}}{1 - i^*_{t}}\frac{a^{g}}{a^{b}}.
\end{equation}

As for $i^*_{t}$ itself, it is determined as follows. If $\lambda_{t}\gamma(0) \leq \frac{a^{b}}{(1+\tau_{t})a^{g}}$, then $i_{t}^{*} = 0$; if $\lambda_{t}\gamma(1) < \frac{a^{b}}{(1+\tau_{t})a^{g}} < \lambda_{t}\gamma(0)$, then $i_{t}^{*}$ is the solution to the following equation in $i$:
\begin{equation}\label{equ-j}
	\lambda_{t}\gamma(i) = \frac{a^{b}}{(1+\tau_{t})a^{g}};
\end{equation} 
if $\lambda_{t}\gamma(1) \geq \frac{a^{b}}{(1+\tau_{t})a^{g}}$, then $i_{t}^{*}=1$.
In other terms,
\begin{equation}\label{jt-eq}
	i^*_{t} = \psi(i^*_{t-1};\tau_{t-1}),
\end{equation}
where 
\begin{equation}\label{jt-eq-1}
		\psi(j;\tau) = 
		\begin{cases}
		0,   &  \text{if} \ \lambda \left(\frac{j}{1 - j}\frac{a^{g}}{a^{b}}\right)\gamma(0) \leq \frac{a^{b}}{(1+\tau)a^{g}}; \\
		\gamma^{-1}(\frac{a^{b}}{(1+\tau)a^{g}\lambda \left(\frac{j}{1 - j}\frac{a^{g}}{a^{b}}\right)}), & \text{if} \ 
		\lambda \left(\frac{j}{1 - j}\frac{a^{g}}{a^{b}}\right)\gamma(1) < \frac{a^{b}}{(1+\tau)a^{g}} < \lambda \left(\frac{j}{1 - j}\frac{a^{g}}{a^{b}}\right)\gamma(0); \\
		1,   &  \text{if} \ \lambda \left(\frac{j}{1 - j}\frac{a^{g}}{a^{b}}\right)\gamma(1) \geq \frac{a^{b}}{(1+\tau)a^{g}}. 
	\end{cases}
\end{equation}
Note that $\psi(j;\tau)$ is continuous and non-decreasing in $j$.

We now assume for simplicity that the government can commit to a constant tax rate, 
\[
\tau_{t} = \tau, \ t=0,1,... .
\]
and consider the discrete dynamic system given by
\begin{equation}\label{j-eq-dyn}
	i^*_{t+1} = \psi(i^*_{t},\tau), \ t = 0,1,\ldots
\end{equation}
and equation (\ref{jt-eq-1}). This system fully describes the equilibrium dynamics in our model for a given tax rate $\tau$.

Let us characterise its main properties. We begin with the characterisation of the fixed point, \[j = \psi(j,\tau), \ t = 0,1,\ldots.\]

\begin{claim}
	If $\lambda (0)\gamma(0) < \frac{a^{b}}{a^{g}}$, then for sufficiently small $\tau$, $j=0$ is a locally asymptotically stable fixed point.
\end{claim}
\begin{proof}
	If $\tau$ is sufficiently small, then $\lambda (0)\gamma(0) < \frac{a^{b}}{(1+\tau)a^{g}}$. Therefore, if $j$ is also sufficiently small, then $\lambda \left(\frac{j}{1 - j}\frac{a^{g}}{a^{b}}\right)\gamma(0) < \frac{a^{b}}{(1+\tau)a^{g}}$ and hence $\psi(j,\tau)=0$.
\end{proof}

The total quantity of the green good in the SSE associated with the fixed point $i^{*}=0$ is zero. The total quantity of the brown good, $B^{*}$, is a solution to the equation 
\begin{equation*}
	B = \mu (B, \tau B) a^{b}l^d,
\end{equation*}
which, using \[
l^d = 1 - l^h = 1 - \frac{\tau B}{\mu (B, \tau B)a^{b}},
\]becomes 

\begin{equation} \label{eqD}
B = \frac{\mu (B, \tau B) a^{b}}{1+\tau}.
\end{equation}

Recall that $\mu(B,H)$ is decreasing in $B$ and increasing in $H$: health status worsens with more of brown good produced and consumed, and it improves with more publicly funded healthcare services provided. Since $H=\tau B$, the sign of the relationship between $\mu(B,\tau B)$ and $B$ is ambiguous: 
\[
\frac{d\mu}{dB} = \frac{\partial{\mu}}{\partial{B}} + \tau \frac{\partial {\mu}}{\partial{H}} =  \frac{\partial{\mu}}{\partial{B}} \left( 1 + \frac{\varepsilon_H}{\varepsilon_B}\right),
\]
where $\varepsilon_B \equiv \frac{B}{\mu}\frac{\partial{\mu}}{\partial{B}}<0$ and $\varepsilon_H \equiv \frac{H}{\mu}\frac{\partial{\mu}}{\partial{H}}>0$ are partial elasticities. In what follows we will assume that 
\begin{equation} \label{elast}
    \varepsilon_H < \lvert \varepsilon_B \rvert ,
\end{equation}
that is, the health status is relatively more sensitive to brown good production and consumption than to the healthcare provision. Assumption (\ref{elast}) is reasonable in the context of our model; otherwise, producing and consuming more of the brown good would be good for health because of the higher spending on healthcare services funded by taxing that good. Also note that under this assumption normalisation $\mu (0, 0) = 1$ implies $\mu (B, \tau B) < 1$ for $B>0$. In this sense, $\mu (B, \tau B)$ is the damage function.

Under assumption (\ref{elast}), $\mu (B, \tau B)$ is decreasing in $B$, and, therefore, equation (\ref{eqD}) has a unique solution, $B^*>0$.

Intuitively, the economy will be trapped in the brown SSE if the preference for green good is weak even for the `greenest' consumer or the technology of green good production is insufficiently developed in comparison to that of the brown good production. Increasing the tax on brown good sales can break this condition. 

\begin{claim}
	For sufficiently large $\tau$, $j=1$ is globally asymptotically stable fixed point.
\end{claim}
\begin{proof}
	It is sufficient to observe that $j=1$ is globally asymptotically stable for any $\tau$ such that $\lambda (\infty)\gamma(1) > \frac{a^{b}}{(1+\tau)a^{g}}$, because for such $\tau$, $\psi(j;\tau) > j \ \forall  j\in[0,1]$.
\end{proof}

The total quantity of the green good in the SSE associated with the fixed point $i^{*}=1$ is equal to $G^{*} = a^{g}$. No brown good is produced in this equilibrium.

Clearly, this is more likely to hold when the preference for green good is strong even for the `least green' consumer or the technology of green good production is sufficiently developed. In addition, this equilibrium is more likely to emerge, the higher is the rate of tax on the brown good.

\begin{claim}
	For a given $\tau \geq 0$ , if $\lambda (0)\gamma(0) < \frac{a^{b}}{(1+\tau)a^{g}}$ and $\lambda (\infty)\gamma(1) > \frac{a^{b}}{(1+\tau)a^{g}}$, then both $j=0$ and $j=1$ are locally asymptotically stable fixed points. 
\end{claim}

In this case, generically, there exists at least one unstable fixed point at some $j \in (0,1)$, and a suitable choice of $\tau$ can put the dynamic system on the path converging to the green equilibrium, associated with the fixed point $j=1$. The following claim says that when both $j=0$ and $j=1$ are locally asymptotically stable fixed points, then under certain conditions the green SSE (associated with $i^{*}=1$) is characterised by a higher consumption level than the `brown' SSE (associated with $i^{*}=0$).

As above, let $B^*$ and $G^*$ denote the total quantities of the brown good and the green good in the SSE associated with fixed points $i^*=0$ and $i^*=1$, respectively. 

\begin{claim}
	Let $\lambda (0)\gamma(0) < \frac{a^{b}}{(1+\tau)a^{g}}$ and $\lambda (\infty)\gamma(1) > \frac{a^{b}}{(1+\tau)a^{g}}$, for a given $\tau\geq 0$. Then $G^*>B^*$ if $a^b \leq a^g$.
\end{claim}

\begin{proof}
    Using the last inequality condition above,
    \[
    \frac{B^*}{G^*} 
    = \frac{\mu(B^*, \tau B^*)a^b}{(1+\tau )a^g} < 1
    \] 
\end{proof}

Because the green and brown goods are perfect substitutes, higher aggregate consumption level in the green SSE suggests that it is socially preferred to the brown SSE.\footnote{Alternatively, one can define the social welfare as the aggregate utility of the households; the welfare comparison gives the same result (see Appendix A for the proof). We prefer the comparison of the consumption levels because in the aggregate utility consumption of households with different preferences for the green good effectively is weighted differently.} Condition $a^b \leq a^g$, which says that the green firm is at least as productive as the brown firm, is, of course, sufficient but not necessary. The aggregate consumption in the green equilibrium is still higher than in the brown equilibrium even if the green sector is less well developed than the brown sector, $a^g < a^b$, but the gap is not too wide.

\section{Green transition policy}
The results presented in the previous section suggest that for any given distribution of individual preferences such that the initial consumption of the brown good is positive the government can move the economy to the green SSE, with zero brown good consumption, instantaneously, by taxing it at a sufficiently high rate. That is, the tax rate must satisfy $\tau \geq  \hat{\tau}_t (1) \equiv \frac{a^b}{a^g \lambda_t \gamma (1)}-1$. This policy, however, ignores the social norm externality in the green good consumption reflected in the social component of household preferences. Since sales tax entails a  negative shock to the households' consumption and welfare levels, a lower tax achieving the same goal would be a better policy. 

Indeed, the government can exploit the social norm externality and impose sales tax at the rate just enough to move the marginal consumer over the threshold, $\tau \geq  \hat{\tau}_t(\hat{j}) \equiv \frac{a^b}{a^g \lambda_t \gamma (\hat{j})}-1$. Here, $\hat{j} = \psi \left( \hat{j};\tau \right)$. The latter condition describes the locally asymptotically unstable SSE or, if there are more than one, then the one associated with the highest $j$. The position of the marginal consumer depends on the value of this $j$ and on the distribution of of $\gamma (i)$. Since, generically, $\hat{j}<1$, we have $\gamma (\hat{j}) > \gamma (1)$ and, therefore, $\hat{\tau}_t(\hat{j}) < \hat{\tau}_t (1)$.

Under this policy, as the green good consumption increases in the initial period, more households will gradually switch to the green good in the subsequent periods, and the economy will converge to the green SSE. Note that in both cases, with the high tax rate $\hat{\tau}_t (1)$ and the moderate tax rate $\hat{\tau}_t(\hat{j})$, the tax can be eventually removed if $\lambda (\infty)\gamma(1) > \frac{a^{b}}{a^{g}}$. Furthermore, the local stability property of the green SSE allows reducing the tax rate along the transition path and still maintaining convergence.\footnote{This result is similar to the one obtained by \citet{borissov2019carbon} in a setting where the existence of human capital externalities enables transition to the green technology by the means of temporary moderate taxation of the brown good.} This is more likely to be the case if the social pressure to buy the green good is relatively high, or the initial preference for green good of the least green consumer is relatively low, or the production technology in the green-good sector is relatively well developed compared to that of the brown-good sector. 

\section{Concluding remarks}
We analysed a green transition policy in the form of a tax on brown good consumption in an economy where individual preferences for green consumption combine a constant individual-specific intrinsic component and an evolving common social component. We characterise the dynamic properties of the equilibria in this economy when the social preference exerts a positive externality in green consumption. This creates complementarity between preferences and policy that discourages brown good consumption and allows exploiting the social norm effect in policy design. We derive the conditions for local stability of the green steady-state equilibrium and show that to achieve green transition the constant tax rate that takes the evolution of social norm into account can be lower than the constant tax rate ignoring the social effect. Moreover, stability ensures that the tax rate can be reduced along the transition path and eventually even be removed, resulting in smaller welfare loss in transition. Thus, a better design of the green transition policy can be achieved when the interaction between policy and preferences is taken into account.

In this paper we made simplifying assumptions about the linearity of production technology and the consumer preferences, to make the dynamic properties of the equilibria tractable. We conjecture that qualitatively similar results will hold for more general, `standard' assumptions on production function and utility function. Rather than calculating the optimal tax policy, we focussed on the conditions under which the tax can facilitate green transition. To construct the optimal path for the tax rate (say, maximising consumption along the transition path) one needs to make more specific assumptions on the distribution of intrinsic green preferences, the shape of the social preferences, and the damage function. Such an exercise seems of a minor value, given the highly stylised nature of the model, adopted to investigate general patterns, but might be of interest, for example, for the purpose of comparison of policy implications for different distributions of preferences.

\bibliographystyle{elsarticle-harv}
\bibliography{references}

\appendix
\section{Social welfare comparison of green and brown steady state equilibria}
 Define the social welfare as the aggregate utility of the households,
\[
SW_t = \int_0^1u^id\Gamma(i),
\] where $\Gamma(i)$ describes the distribution of individual-specific preference parameter $\gamma(i)$ across consumers.

The following claim says that when both $j=0$ and $j=1$ are locally asymptotically stable fixed points, then the green steady-state equilibrium (associated with $i^{*}=1$) is socially preferable to the `brown' steady-state equilibrium (associated with $i^{*}=0$).

As above, let $B^*$ and $G^*$ denote the total quantities of the brown good and the green good in the steady-state equilibria associated with fixed points $i^*=0$ and $i^*=1$, respectively, and let $SW_B^*$ and $SW_G^*$ be the corresponding social welfare levels. 

\begin{claim}
	Let $\lambda (0)\gamma(0) < \frac{a^{b}}{(1+\tau)a^{g}}$ and $\lambda (\infty)\gamma(1) > \frac{a^{b}}{(1+\tau)a^{g}}$, for a given $\tau\geq 0$. Then $SW_G^*>SW_B^*$.
\end{claim}

\begin{proof}
    The social welfare levels in the brown and green steady-state equilibria are given by  $SW_B^* = B^*$ and $SW_G^* = \lambda (\infty)\int_0^1\gamma(i)d\Gamma(i)G^* = \lambda (\infty)\hat{\gamma}C*$, where $\hat{\gamma} \equiv \int_0^1\gamma(i)d\Gamma(i)$ is the preference parameter of the average household. Hence, using the second inequality condition above,
    \[
    \frac{SW_B^*}{SW_B^*} = \frac{B^*}{\lambda (\infty)\hat{\gamma}G^*} 
    = \frac{\mu(B^*, \tau B^*)}{(1+\tau )\lambda (\infty)\hat{\gamma}G^*}  
    < \mu(B^*, \tau B^*) \frac{\gamma(1)}{\hat{\gamma}}<1
    \] 
\end{proof}
\end{document}